\documentclass[aip,jmp,reprint,floatfix,superscriptaddress,onecolumn]{revtex4-1}
\pdfoutput=1
\usepackage[ascii]{inputenc}
\usepackage[bookmarksnumbered,hypertexnames=false]{hyperref}
\usepackage[T1]{fontenc}
\usepackage{lmodern}
\usepackage[USenglish]{babel}
\usepackage{graphicx,amsmath,amssymb,amsmath,amsthm,amssymb,fullpage,color,bbm,braket,microtype,mathrsfs,parskip,latexsym,multirow,mathtools}
\usepackage[mathscr]{eucal}
\newtheoremstyle{custom}{3pt}{3pt}{\slshape}{}{\bfseries}{.}{ }{}
\theoremstyle{custom}
\newtheorem{theorem}{Theorem}[section]
\newtheorem{proposition}[theorem]{Proposition}
\newtheorem{lemma}[theorem]{Lemma}

\newtheorem{conjecture}[theorem]{Conjecture}
\theoremstyle{definition}

\theoremstyle{remark}
\newtheorem{remark}[theorem]{Remark}
\DeclareMathOperator{\qcap}{qcap}
\DeclareMathOperator{\Qcap}{Qcap}
\DeclareMathOperator{\QMC}{QMC}
\DeclareMathOperator{\QMF}{QMF}
\DeclareMathOperator{\IMM}{IMM}
\DeclareMathOperator{\TNS}{TNS}
\DeclareMathOperator{\GL}{GL}

\DeclareMathOperator{\PGL}{PGL}
\DeclareMathOperator{\tspan}{span}
\DeclareMathOperator{\ttrace}{trace}
\DeclareMathOperator{\trank}{rank}
\DeclareMathOperator{\coef}{coeff}
\DeclareMathOperator{\sgn}{sign}
\DeclareMathOperator{\Id}{Id}
\newcommand{\odd}{\mathit{odd}}
\newcommand{\even}{\mathit{even}}

\newcommand{\FS}{\mathfrak S}
\newcommand{\BZ}{\mathbb Z}\newcommand{\ZZ}{\mathbb Z}
\newcommand{\BC}{\mathbb C}\newcommand{\CC}{\mathbb C}\newcommand{\bbC}{\mathbb{C}}
\newcommand{\BP}{\mathbb P}\newcommand{\pp}[1]{\mathbb P^{#1}}\newcommand{\bbP}{\mathbb{P}}
\newcommand{\cS}{\mathcal S}
\newcommand{\ot}{\otimes}
\newcommand{\op}{\oplus}
\newcommand{\ra}{\to}
\newcommand{\otc}{\otimes\cdots\otimes}
\newcommand{\La}[1]{\Lambda^{#1}}
\newcommand{\hd}{,\dots,}

\renewcommand{\b}{\beta}
\newcommand{\s}{\sigma}
\newcommand{\vvirg}{, \dots ,}
\newcommand{\eps}{\varepsilon}
\newcommand{\udelta}{{\underline{\delta}}}
\newcommand{\ueps}{{\underline{\varepsilon}}}
\newcommand{\ueta}{{\underline{\eta}}}
\newcommand{\bfK}{\mathbf{K}}

\begin{document}
\title{\texorpdfstring{Matrix product states and the \\ quantum max-flow/min-cut conjectures}{Matrix product states and the quantum max-flow/min-cut conjectures}}
\author{Fulvio Gesmundo}
\affiliation{QMATH, University of Copenhagen, Universitetsparken 5, 2100 K\o{}benhavn \O{}, Denmark}
\email{fulges@math.ku.dk}
\author{J.~M. Landsberg}
\affiliation{Department of Mathematics, Texas A\&M University, College Station, TX 77843-3368, USA}
\email{jml@math.tamu.edu}
\author{Michael Walter}
\affiliation{QuSoft, Korteweg-de Vries Institute for Mathematics, Institute for Theoretical Physics, Institute for Logic, Language and Computation, University of Amsterdam, 1098 XG Amsterdam, Netherlands}
\affiliation{Stanford Institute for Theoretical Physics, Stanford University, Stanford, CA 94305, USA}
\email{m.walter@uva.nl}
\hypersetup{pdftitle={Matrix product states and the quantum max-flow/min-cut conjectures}, pdfauthor={Fulvio Gesmundo, J. M. Landsberg, Michael Walter}}

\begin{abstract}
In this note we discuss the geometry of matrix product states with periodic boundary conditions and provide three infinite sequences of examples where the quantum max-flow is strictly less than the quantum min-cut.
In the first we fix the underlying graph to be a $4$-cycle and verify a prediction of Hastings that inequality occurs for infinitely many bond dimensions.
In the second we generalize this result to a $2d$-cycle.
In the third we show that the $2d$-cycle with periodic boundary conditions gives inequality for all $d$ when all bond dimensions equal two, namely a gap of at least~$2^{d-2}$ between the quantum max-flow and the quantum min-cut.
\end{abstract}
\maketitle

\section{Introduction}

A {\it tensor network} associated to a graph is a way of constructing tensors in large spaces from smaller building-block tensors.
From the perspective of algebraic geometry, tensor networks provide a natural way of constructing varieties of tensors.
The use of graphs to study tensors dates back at least to Clifford in 1881 (see Ref.~\onlinecite[Fig.~2.11.1]{MR2865915}).
In applied mathematics and physics, tensor networks are versatile tools for efficiently approximating high-dimensional data, such as ground states of many-body quantum systems in condensed matter physics (see, e.g., Refs.~\onlinecite{MR2535056,verstraete2009tnsreview,orus2014tnsreview}).

In this paper we focus exclusively on {\it translation-invariant matrix product states with periodic boundary conditions} as defined in \S\ref{mpsdef}.
We use them to obtain three sequences of examples of tensor networks where the quantum max-flow (\S\ref{subsec:qmf}) is strictly less than the quantum min-cut (\S\ref{subsec:qmc}) for certain partitions.
Our results disprove a natural conjecture and verify a numerical prediction from Ref.~\onlinecite{MR3613509} (\S\ref{subsec:results}).

\subsection{Matrix product states with periodic boundary conditions}\label{mpsdef}
In this paper, we work exclusively with tensor networks associated with the oriented cyclic graphs~$C_m$, as in Figure~\ref{fig:circletns}~(a), and one building-block tensor~$T$ taken from $\BC^N\ot \BC^n\ot (\BC^n)^*$, which we write as $A\ot B\ot B^*$.

\begin{figure}[!htb]
\begin{center}
\raisebox{2cm}{(a)}
\includegraphics{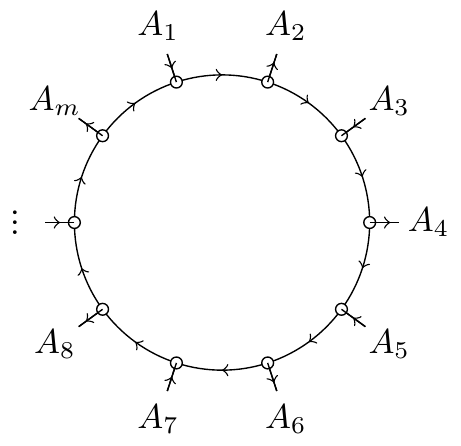}
\qquad\raisebox{2cm}{(b)}
\raisebox{-.1cm}{\includegraphics{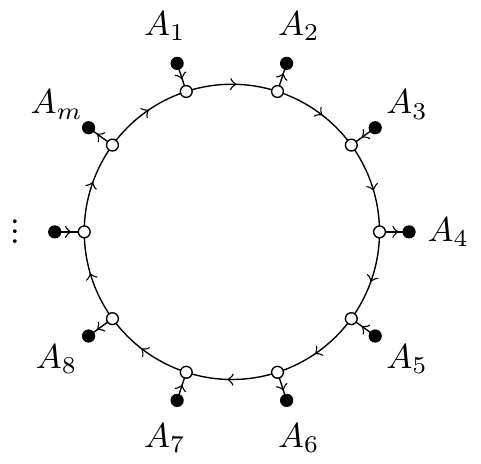}}
\end{center}
\caption{(a) Cyclic graph $C_m$ with $m$~vertices and clockwise orientation.
This graph is used for constructing a matrix product state.
The $m$~copies of the $A$ vector space are labeled $A_1,\dots,A_m$.
Arrows on the edges on the circle indicate $B$ (outgoing) vs.\ $B^*$ (incoming).
Arrows on the external edges indicate partition into sources $S$ (incoming) and sinks $\bar S$ (outgoing).
In our case the quantum capacities are $n$ for edges on the circle and $N$ for external edges.
(b) Extented graph $\widehat C_m$ obtained by adding terminal vertices (black) at the dangling ends of the external edges.}
\label{fig:circletns}
\end{figure}

One constructs tensors in~$A^{\ot m}=(\BC^N)^{\ot m}$ from this set-up as follows:
For each vertex, we associate a copy of $A$ to the external (physical) edge, a copy of $B$ to the outgoing edge from the circle, and a copy of $B^*$ to the incoming edge in the circle.
We thus obtain a tensor $\Phi(T)\in A^{\ot m}$ by
\begin{align*}
\Phi(T)=\operatorname{contract}(T^{\ot m}),
\end{align*}
where we place a copy of~$T$ at each vertex, and then contract $B\ot B^*\ra \BC$ along the edges on the circle.
We will refer to the tensor $\Phi(T)\in A^{\ot m}$ as the {\it tensor network state} associated to~$(C_m,T)$.
The dimensions $N,n$ are known as the \emph{bond dimensions}.

In physics, tensors such as $\Phi(T)$ are also known as \emph{translation-invariant matrix product states with periodic boundary conditions}.
In indices, if $a_1\hd a_N$ is a basis of $A$, $b_1\hd b_n$ a basis of $B$ with dual basis $\b^1\hd \b^n$, so we may write $T=\sum_{i,s,t}T^{is}_{\;\;\;\,t} a_i\ot b_s\ot\b^t$, then
\begin{align}\label{eq:contraction in coordinates}
\Phi(T)=\sum T^{i_1s_1}_{\;\;\;\;\;\;\;s_m} T^{i_2s_2}_{\;\;\;\;\;\;\;s_1}\cdots T^{i_ms_m}_{\;\;\;\;\;\;\;s_{m-1}} a_{i_1}\otc a_{i_m}.
\end{align}
If we think of each $T^i=(T^{is}_{\;\;\;\,t})$ as an $n\times n$ matrix, we may write the $a_{i_1}\otc a_{i_m}$ coefficient of $\Phi(T)$ as
$$
 \ttrace(T^{i_m}\cdots T^{i_1}),
$$
which explains the terminology ``matrix product state''.
``Periodic boundary conditions'' refers to the fact that the graph is a cycle and not a linear graph, and ``translation-invariant'' means that the same tensor~$T$ is placed at each vertex.
This implies a cyclic $\ZZ_m$-symmetry for $\Phi(T)$ that will play an important role in this paper.

Let $\TNS(C_m,N,n)\subset A^{\ot m}$ denote the set of all tensor network states $\Phi(T)$ associated to $C_m$, and let
$\BP\overline{\TNS}(C_m,N,n)\subset \BP(A^{\ot m})$ denote its Zariski closure in projective space.
In other words, $\Phi$ defines a rational map
\begin{align*}
\Phi: \BP (A\ot B\ot B^*) \dashrightarrow \BP(A^{\ot m})
\end{align*}
(which we denote by the same symbol~$\Phi$), and $\BP \overline{\TNS}(C_m,N,n)$ is the closure of its image.
Note that the image is linearly degenerate, lying in the space of invariants $(A^{\ot m})^{\BZ_m}$, where $\BZ_m$ acts by cyclically permuting the vertices.
Also note that the group~$\GL(B)$, the invertible linear maps $B\ra B$, acts on~$B$ and $B^*$; so $\PGL(B)=\GL(B)/\{\BC^*\Id\}$ acts on $\BP(A\ot B\ot B^*)$, preserving the fibers of~$\Phi$.
Thus the generic fiber of $\Phi$ contains a space isomorphic to $\PGL(B)$.

\subsection{Quantum max-flow}\label{subsec:qmf}
Now partition the external edges into two sets~$S$, $\bar S$ (`sources' and `sinks', see Figure~\ref{fig:circletns}).
This induces a splitting $A^{\ot m}=(\ot_{s\in S}A_s) \ot (\ot_{t\in\bar S}A_t)$.
Given $\Phi(T)\in A^{\ot m}$, we get an induced linear map
\begin{align*}
  \Phi(T)_{S,\bar S}: \bigotimes_{s\in S}A_s^* \to \bigotimes_{t\in\bar S}A_t
\end{align*}
(called a {\it flattening} in the geometry literature).
The {\it quantum max-flow} is defined as the maximal rank over all $\Phi(T)$ of this flattening~\cite{MR3513725}
\begin{align*}
  \QMF'(C_m,(S,\bar S), N,n):=\max_{T\in A\ot B\ot B^*} \trank\left(\Phi(T)_{S,\bar S}\right).
\end{align*}
The prime reminds us that we place the same tensor~$T$ at each vertex.
When $N=n$, we suppress it from the notation and just write $\QMF'(C_m,(S,\bar S), N)$.

\subsection{Quantum min-cut}\label{subsec:qmc}
To define the quantum min-cut, it is useful to define an extended graph $\widehat C_m$ by adding terminal vertices at the dangling ends of the external edges, as in Figure~\ref{fig:circletns},~(b).
Now we may also think of the partition $(S,\bar S)$ as a partition of the \emph{terminals} into two sets.
A  \emph{cut} in this situation is a partition of the vertices of $\widehat C_m$ into two sets $(\cS$, $\bar\cS)$, with $S\subseteq\cS$ and $\bar S\subseteq\bar\cS$.
Define the {\it quantum capacity} of a cut $(\cS,\bar\cS)$ by
$$
\Qcap(\cS,\bar\cS):=\prod_{v\in \cS,\ w\in \bar\cS,\ vw\in E} \qcap(vw),
$$
where the quantum capacity of an edge, $\qcap(vw)$, is the dimension of the vector space associated to it.
In our case the quantum capacities are $n$ for edges on the circle and $N$ for external edges.
Following Ref.~\onlinecite{MR3513725}, define the \emph{quantum min-cut} by
\begin{align*}
  \QMC(C_m,(S,\bar S),N,n):=\min_{(\cS,\bar\cS)} \Qcap(\cS,\bar\cS),
\end{align*}
where we minimize over all cuts $(\cS,\bar\cS)$ for $(S,\bar S)$.
Again we write $\QMC(C_m,(S,\bar S),N)$ when $N=n$.

\subsection{Quantum max-flow vs.\ quantum min-cut}\label{subsec:results}
In Ref.~\onlinecite{MR3513725} the authors propose that tensor networks with physical edges divided into two sets can be viewed as ``transporting'' linear-algebraic quantities such as rank and entanglement, and are properly viewed as quantum analogs of graphs modeling flow networks.
In the classical case, that is, for flow networks, it is well-known that the maximal flow passing from sources to sinks is equal to the minimum cut separating the sources from sinks---this is the famous \emph{max-flow min-cut theorem}.
In the quantum case, it is well-known and easy to see that
\begin{align}\label{eq:weak duality}
  \QMF'(C_m,(S,\bar S),(N,n)) \leq \QMC(C_m,(S,\bar S),(N,n))
\end{align}
and similarly for arbitrary graphs.
Indeed, any cut $(\cS,\bar\cS)$ induces a factorization of the linear map $\Phi(T)_{S,\bar S}$ through a vector space of dimension equal to the quantum capacity of the cut, which implies~\eqref{eq:weak duality} at once.

In Ref.~\onlinecite{MR3513725}, the authors studied to what extent \emph{equality} holds in the quantum case, prompted by the ``quantum max-flow/min-cut conjecture'' from Ref.~\onlinecite[Conjecture C.1]{calegari2010positivity}.
Note that in the quantum case, the quantum min-cut is still straight-forward to compute (it can be readily reduced to computing a classical min-cut), but the quantum max-flow may be difficult to compute directly.

The original conjecture was vastly more general than the set-up here, but it in particular implied that, for all $(S,\bar S)$ and $N$, $\QMF'(C_m,(S,\bar S),N)=\QMC(C_m,(S,\bar S),N)$.
This was shown to be false in Ref.~\onlinecite{MR3513725}, namely it was proved that $\QMF'(C_4,(S,\bar S), 2) = 3 < 4 = \QMC(C_4, (S,\bar S), 2)$, where $S=\{1,3\}$, $\bar S=\{2,4\}$.
In Ref.~\onlinecite{MR3513725}, the cyclic graph $C_4$ was depicted as in Figure~\ref{fig:mod4},~(a), where the vertices in~$S$ are on the left and vertices in~$\bar S$ on the right.

\begin{figure}
\begin{center}
\raisebox{1.4cm}{(a)}
\includegraphics{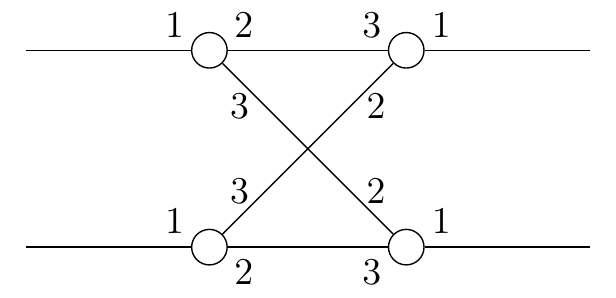}
\qquad\raisebox{1.4cm}{(b)}
\includegraphics{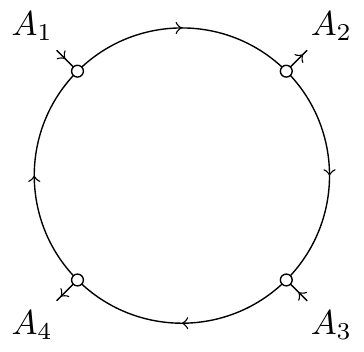}
\end{center}
\caption{(a) Cyclic graph $C_4$ as presented in Ref.~\onlinecite{MR3513725}. The labels $1$,$2$,$3$ refers to the spaces $A$,$B$,$B^*$, respectively.
(b) The same graph, but presented as in Figure~\ref{fig:circletns}. The cyclic $\ZZ_4$-symmetry is obvious.}
\label{fig:mod4}
\end{figure}

This raised the question whether weaker forms of a quantum max-flow/min-cut conjecture might be true, such as the following (stated for simplicity only for cyclic graphs):

\begin{conjecture}\label{conjecture:weak}
  For any partition~$(S,\bar S)$ of the external edges, there exists $N_0$ such that, for all $N\geq N_0$,
  \[ \QMF'(C_m, (S, \bar S), N) = \QMC(C_m, (S,\bar S), N). \]
\end{conjecture}

We will be particularly interested in the case $m=2d$, where $S$ corresponds to the odd-labeled indices and $\bar S$ to the even-labeled indices, as in Figure~\ref{fig:circletns},~(b).
We will write $(S,\bar S)=(\odd,\even)$ accordingly.
The example from Ref.~\onlinecite{MR3513725} is of this form, see Figure~\ref{fig:mod4},~(b).
Numerical evidence reported in Ref.~\onlinecite{MR3613509} suggested a cyclic dependency of $\QMF'(C_4, (\odd,\even), N)$ on $N \bmod 4$ in the situation of Figure~\ref{fig:mod4}, and, therefore, that Conjecture~\ref{conjecture:weak} is false.
Our first result proves that this is indeed the case:

\begin{theorem}\label{thm:mod4}
  For all $N$,
  \[
  \QMF'(C_4,(\odd,\even),N) \leq \begin{cases}
  N^2 & \text{if $N \equiv 0,1 \pmod 4$} \\
  N^2 - 1 & \text{if $N \equiv 2,3 \pmod 4$},
  \end{cases}
  \]
  while $\QMC(C_4,(\odd,\even),N)=N^2$ for all $N$.
  Moreover, equality holds for all square bond dimensions $N=k^2$, so in an infinite number of cases where $N\equiv0,1\pmod4$.
\end{theorem}

Numerical evidence suggests that our bound is tight.\cite{MR3613509}
For $m > 4$, numerical evidence is much harder to obtain.
Yet we provide a partial generalization of Theorem~\ref{thm:mod4} to higher cycles:

\begin{theorem}\label{thm:mod4higher}
  For all $d\equiv2\pmod4$ and $N$,
  \[
  \QMF'(C_{2d},(\odd,\even),N) \leq \begin{cases}
  N^d & \text{if $N \not\equiv 3\pmod 4$} \\
  N^d - 1 & \text{if $N\equiv 3\pmod 4$},
  \end{cases}
  \]
  while $\QMC(C_{2d},(\odd,\even),N)=N^d$ for all $N$.
  Again we have equality for all $N=k^2$.
\end{theorem}

We prove Theorems~\ref{thm:mod4} and \ref{thm:mod4higher} in \S\ref{sec:proof mod4}.

\begin{remark}
As suggested in Ref.~\onlinecite{MR3613509}, one might instead consider a weaker version of Conjecture~\ref{conjecture:weak}, where we only demand that equality holds for some (and hence infinitely many) $N>1$.
As evidence, Ref.~\onlinecite{MR3613509} proved that for all $G$, $(S,\bar S)$, as functions of $N$, we have $\QMF'(G,(S,\bar S),N)=\QMC(G,(S,\bar S),N)(1-o(1))$ (cf.~Refs.~\onlinecite{hayden2016holographic,nezami2016multipartite}, where a similar result was proved in the scenario where we place different tensors at each site).
\end{remark}


For $N=2$, our bound in Theorem~\ref{thm:mod4higher} can be improved.
This is shown by our next result, which gives an infinite sequence of graphs with constant bond dimension for which the quantum max-flow is strictly smaller than the quantum min-cut.

\begin{theorem}\label{thm:2d}
  For all $d\geq2$,
  \begin{align*}
    \QMF'(C_{2d},(\odd,\even),2) \leq \frac34 2^d < 2^d = \QMC(C_{2d},(\odd,\even), 2).
  \end{align*}
\end{theorem}
Theorem~\ref{thm:2d} puts the observation in Ref.~\onlinecite{MR3513725} that $\QMF'(C_6,(\odd,\even),2)=6<8$ into a general context.
We prove Theorem~\ref{thm:2d} in \S\ref{sec:Nn2}.
Interestingly, the rank defect is of a \emph{nonlinear} origin, unlike in our preceding theorems.
Numerical calculations up to~$d=10$ suggest that our bound is tight.

\begin{remark}\label{rem:immrem}
If $N=n=k^2$ for some integer $k$, then $\overline{\TNS}(C_m, k^2, k^2)$ consists of the diagonal degenerations of the $m$-times iterated $k\times k$-matrix multiplication tensor
\[
\IMM^m_k:=\sum_{i_p =1}^k (a_1)^{i_1}_{i_2} \otimes (a_2)^{i_2}_{i_3} \otimes \cdots \otimes (a_{m-1})^{i_{m-1}}_{i_m}\otimes (a_m)^{i_{m}}_{i_1},
\]
where $\{(a_p)^i_j\}_{ij = 1 \vvirg k}$ is a basis of the $p$-th copy of $A$. The variety $TNS(C_m,k^2,k^2)$ is the closure of the image via $\Phi$ of the $GL(A)$-orbit of the tensor $T = \sum a^i_j \otimes \beta^j_k \otimes b^k_i$, where $\{ a^i_j \}_{ij = 1 \vvirg k}$ is a basis of $A$ and $\{  b^i_j \}_{ij = 1 \vvirg k}$ is a basis of $B$ with dual basis $\{ \beta^j_i \}_{ij = 1 \vvirg k}$. In fact, $T = \IMM^3_k$ regarded as an element of $A \otimes B^* \otimes B$ and $\Phi(T) = \IMM^m_k$.

The ranks of the flattenings of $\IMM^m_k$ are known (see, e.g., Refs.~\onlinecite{MR3513064,BuhChrZui:NondetQuantumComCompl}):
when $m=2d$ and $S=\{1\hd d\}$, $\bar S=\{d+1\hd 2d\}$, the quantum min-cut is $k^4$ and indeed the flattening is of maximal rank~$k^4$.
Similarly, in the case that $(S,\bar S)=(\odd,\even)$, the quantum min-cut as well as the rank of the flattening are equal to~$k^{2d}$.
Thus:
\begin{align*}
  \QMF'(C_{2d},(\odd,\even),k^2)=k^{2d}=\QMC(C_{2d},(\odd,\even),k^2)
\end{align*}
for all~$k$.
This shows the equality statements in Theorems~\ref{thm:mod4} and \ref{thm:mod4higher}.
\end{remark}

\subsection{Notation and conventions}
$A$, $B$ are complex vector spaces respectively of dimensions $N,n$.
$\GL(A)$ denotes the group of invertible linear maps $A\ra A$ and $\FS_d$ denotes the permutation group on $d$ elements.
We denote the elements of $\ZZ_d$, the cyclic group of order $d$, by $[k]$ for $k\in\ZZ$.
We write $S^mA$ and $\La mA$ for the symmetric and antisymmetric subspaces of $A^{\ot m}$, respectively.

\section{Proof of Theorems~\ref{thm:mod4} and \ref{thm:mod4higher}}\label{sec:proof mod4}

We already showed in Remark~\ref{rem:immrem} that we have equality in all square dimensions~$N=k^2$, so we only need to establish the upper bound.
We first prove Theorem~\ref{thm:mod4}, which rigorously establish the defects observed in Ref.~\onlinecite{MR3613509} and Theorem~\ref{thm:mod4higher}.

\subsection{Proof of Theorem~\ref{thm:mod4}}

It is clear from the cyclic symmetry of the tensor network in Figure~\ref{fig:mod4},~(b) that $\Phi(T)\in A^{\ot 4}$ has a cyclic $\ZZ_4$-symmetry, generated by $\pi=(1~2~3~4)$.
Thus we need to understand the invariant subspace $(A^{\ot4})^{\ZZ_4}$.

For this, we order the tensor factors of $A^{\ot4}$ as $(A_1\ot A_3)\ot(A_2\ot A_4)$, corresponding to the flattening~$\Phi(T)_{\odd,\even}$ of interest.
We claim that
\begin{align}\label{eq:mod4 symmetries}
  (A^{\ot4})^{\ZZ_4}
\subseteq S^2(S^2A) \op \La2(\La2A)
\subseteq S^2A\ot S^2A \op \La2A\ot\La2A
\subseteq (A_1\ot A_3)\ot(A_2\ot A_4).
\end{align}
(In fact, the first inclusion is an equality.)
To see this, recall that the $\ZZ_4$-symmetry is generated by~$\pi=(1~2~3~4)$.
Since $\pi^2=(1~3)(2~4)$, it is clear that $(A^{\ot4})^{\ZZ_4} \subseteq S^2A\ot S^2A \op \La2A\ot\La2A$.
Now note that $\pi$ acts block diagonally with respect to the direct sum.
In fact, $\pi=(1 3)\tau$, where $\tau=(1 2)(3 4)$, so $\pi$ and $\tau$ have the same action on $S^2A\ot S^2A$, while $\pi$ acts by $-\tau$ on $\La2A\ot\La2A$.
But $\tau$ interchanges $A_1\ot A_3$ with $A_2\ot A_4$, so it follows that the $\ZZ_4$-invariant subspace lives in $S^2(S^2A) \op \La2(\La2A)$.
This establishes~\eqref{eq:mod4 symmetries}.

Now, \eqref{eq:mod4 symmetries} implies that the linear map~$\Phi(T)_{\odd,\even}$ is block diagonal, mapping the symmetric subspace of $A_1^*\ot A_3^*$ to the symmetric subspace of $A_2\ot A_4$, and the anti-symmetric subspace to the anti-symmetric subspace.
Moreover, the first block is given by a symmetric matrix, while the second block is given by a skew-symmetric matrix.
Since the rank of a skew-symmetric matrix is always even, while $\dim\La2A=\frac12N(N-1)$ is even if and only if $N\equiv0,1\pmod4$, we find that
\begin{equation}\label{eq:upper bound}
  \QMF'(C_4,(\odd,\even),N)
\leq \begin{cases}
    N^2 & \text{if $N \equiv 0,1 \pmod 4$,} \\
    N^2 - 1 & \text{if $N \equiv 2,3 \pmod 4$.}
  \end{cases}
\end{equation}
(The $N=2$ case will be re-proved geometrically in \S\ref{sec:Nn2}. It also follows by a direct computation, as was done
in Ref.~\onlinecite{MR3513725}.)\qed

\subsection{Proof of Theorem~\ref{thm:mod4higher}}
Now consider a general graph~$C_{2d}$ with $d\equiv2\pmod4$.
Again, $\Phi(T)\in A^{\ot2d}$ has a $\ZZ_{2d}$-symmetry, generated by $\pi=(1~2~3~\cdots~2d)$, so we focus on the $\ZZ_{2d}$-invariant subspace of $A^{\ot 2d}$.

We first note that $\pi^2=(1~3~\cdots~2d-1)(2~4~\cdots~2d)$.
Both $(1~3~\cdots~2d-1)$ and $(2~4~\cdots~2d)$ are $d$-cycles, permuting the odd and even subsystems, respectively.
We can decompose $A^{\ot d} = \bigoplus_z V_z$ into the eigenspaces of such a $d$-cycle, where $z$ runs over the $d$-th roots of unity.
Thus the invariance by $\pi^2$ implies that
\begin{align*}
  (A^{\ot 2d})^{\ZZ_{2d}} \subseteq \bigoplus_z V_z \ot V_{\bar z},
\end{align*}
where as before we use the odd-even ordering of tensor factors.
Since $d$ is even, $z=-1$ is a possible eigenvalue.
Next, note that $\pi=(1~3~\cdots~2d-1)\tau=\tau(2~4~\cdots~2d)$, where $\tau=(1 2)(3 4)\cdots (2d\!\!-\!\!1~2d)$ interchanges the odd and even subsystems.
It follows that $\pi$ acts by $\tau$ on $V_1\ot V_1$, by $-\tau$ on $V_{-1}\ot V_{-1}$, and by $\begin{psmallmatrix}0&\bar z\tau\\z\tau&0\end{psmallmatrix}$ on $V_z\ot V_{\bar z}\op V_{\bar z}\ot V_z$.
In particular,
\begin{align}\label{eq:mod2d symmetries}
  (A^{\ot 2d})^{\ZZ_{2d}} \subseteq S^2(V_1) \op \La2(V_{-1}) \op \bigoplus_{z\neq\pm1} V_z \ot V_{\bar z}.
\end{align}
(For $d=2$, the only eigenspaces are $V_1=S^2A$ and $V_{-1}=\La2A$, so~\eqref{eq:mod2d symmetries} reduces to~\eqref{eq:mod4 symmetries}.)

As before, \eqref{eq:mod2d symmetries} implies that $\Phi(T)_{\odd,\even}$ is block diagonal (e.g., with respect to the three direct summands) and the block that maps $V_{-1}^*$ to $V_{-1}$ is given by a skew-symmetric matrix.
We now compute the dimension of the eigenspace~$V_{-1}$.

\begin{lemma}\label{lem:mults}
  Let $d=2(2a+1)$.
  Then, $\dim V_{-1} = \frac1{2a+1} \sum_{b=1}^{2a+1} \binom {N^{\gcd(2a+1,b)}} 2$.
  In particular, $\dim V_{-1}$ is odd if $N\equiv3\pmod4$.
\end{lemma}
\begin{proof}
Consider the representation of $\ZZ_d$ on $V=(\CC^N)^{\ot d}$ by cyclically shifting tensor factors.
It is straightforward to evaluate its character~$\chi\colon\ZZ_d\to\CC$ in the standard product basis, $a_{i_{[1]}}\otc a_{i_{[d]}}$, where we label the indices by $\ZZ_d$ rather than $\{1,\dots,d\}$ so that it is straightforward to implement the shift:
For all $k\in \{1\hd d\}$,
\begin{align*}
  \chi([k])
= \sum_{i_{[1]},\dots,i_{[d]}\in\{1,\dots,N\}} \prod_{j=1}^d \delta_{i_{[j]}, i_{[j+k]}}
= N^{\lvert\ZZ_d/k\ZZ_d\rvert}
= N^{\gcd(d, k)},
\end{align*}
since we have one free index to choose per orbit of the shift by~$[k]$.
The dimension of the $-1$~eigenspace is the multiplicity of the sign representation in $V$, so given by the normalized inner product
\begin{equation}\label{eq:goal}
  \dim V_{-1} = \frac 1d \sum_{k=1}^d (-1)^k \chi([k]) = \frac1d \sum_{k=1}^d (-1)^k N^{\gcd(d,k)}.
\end{equation}
Using $d=2(2a+1)$, we can calculate the numerator as
\begin{align*}
  \sum_{k=1}^d (-1)^k N^{\gcd(d,k)}
= \sum_{b=1}^{2a+1} N^{\gcd(d,2b)} - N^{\gcd(d,2b-1)}
= \sum_{b=1}^{2a+1} N^{2\gcd(2a+1,b)} - N^{\gcd(2a+1,2b-1)} \\
= \sum_{b=1}^{2a+1} N^{2\gcd(2a+1,b)} - N^{\gcd(2a+1,b)}
= \sum_{b=1}^{2a+1} N^{\gcd(2a+1,b)} \left( N^{\gcd(2a+1,b)} - 1 \right).
\end{align*}
In the third step, we substituted $2b-1$ by $b$ (since this defines a bijection of $\ZZ_{2a+1}$ and $\gcd(2a+1,b)$ only depends on $b$ modulo $2a+1$, the sum is left unchanged).
Thus the multiplicity~\eqref{eq:goal} is given by
\begin{align*}
  \dim V_{-1}
= \frac1d \sum_{b=1}^{2a+1} N^{\gcd(2a+1,b)} \left( N^{\gcd(2a+1,b)} - 1 \right)
= \frac1{2a+1} \sum_{b=1}^{2a+1} \binom {N^{\gcd(2a+1,b)}} 2
\end{align*}
This establishes the first claim.
For the second, assume that $N\equiv3\pmod4$.
In order to prove that $\dim V_{-1}$ is odd, it suffices to show that each binomial coefficient is odd.
But indeed, since $2a+1$ is odd, so is $\gcd(2a+1,b)$.
This implies that $N^{\gcd(2a+1,b)}\equiv 3\pmod4$ for all $b$, which precisely ensures that the binomial coefficients are odd.
\end{proof}

As before, Lemma~\ref{lem:mults} implies that $\Phi(T)_{\odd,\even}$ has a rank defect if $N\equiv3\pmod4$; this establishes Theorem~\ref{thm:mod4higher}. \qed

\begin{remark}
Numerical experiments suggest that there is a rank defect for all even $d$, not just for $d\equiv2\pmod4$.
\end{remark}

\section{\texorpdfstring{The variety of tensor network states $\BP\overline{\TNS}(C_m,N,n)$}{The variety of tensor network states PTNS(C\_m,N,e)}}\label{sec:geometry}
In this section we discuss some general features of the variety~$\BP\overline{\TNS}(C_m,N,n)$.
In particular, we show that, for $N,n\geq m$, the smallest linear subspace containing the variety is~$(A^{\ot m})^{\ZZ_m}$.

Let $v_m(\BP A):=\{ [p] \mid p=\ell^{\otimes m} \text{ for some } 0\neq\ell\in A\}\subseteq\BP(S^mA)$ denote the {\it Veronese variety} of $m$-th powers of linear forms.
For any variety~$X\subseteq \BP V$, let
\begin{align*}
\s_r(X):=\overline{\bigcup_{x_1\hd x_r\in X}\tspan\{x_1\hd x_r\}}
\end{align*}
denote the $r$-th {\it secant variety} of $X$, so $\s_r(v_m(\BP A))\subseteq \BP(S^mA)$ is the Zariski closure of the set of homogeneous polynomials of degree~$m$ in $N$~variables that may be written as the sum of $r$ $m$-th powers of linear forms.

Observe that $\BP\overline{\TNS}(C_m,N,n)$ contains the variety $\sigma_n(v_m(\BP A))\subseteq\BP(S^mA)$.
Indeed, consider the tensor
\begin{align*}
T=a_1\ot b_1\ot\b^1+\dots+a_n\ot b_n\ot\b^n,
\end{align*}
where $b_1\hd b_n$ is a basis of $B$, with dual basis $\b^1\hd \b^n$.
Then~\eqref{eq:contraction in coordinates} implies
\begin{align*}
\Phi(T) = a_1^{\ot m}+\dots+a_n^{\ot m}.
\end{align*}
If the~$a_i$ are chosen as general points of $A$, this projectivizes to a general point of $\s_n(v_m(\BP A))$, and by $\GL(A)$-invariance of the image, the whole variety must be contained in $\BP\overline{\TNS}(C_m,N,n)$.

Not every $\Phi(T)$ is contained in $S^mA$ (outside of the trivial case $n=1$ which we exclude from consideration).
E.g., for $m=2d$, consider the tensor
\begin{align*}
  T=a_1\ot b_1\ot\b^2 + a_2\ot b_2\ot\b^1.
\end{align*}
Then, with respect to the odd-even ordering of tensor factors,
\begin{align*}
  \Phi(T) = a_1^{\ot d} \ot a_2^{\ot d} + a_2^{\ot d} \ot a_1^{\ot d}
\end{align*}
which is not in~$S^m A$.

Now assume that $N,n\geq m$ and consider the tensor
\begin{align*}
  T = a_1\ot b_1\ot\b^m+a_2\ot b_2\ot\b^1+\dots+a_{m-1}\ot b_{m-1}\ot\b^{m-2}+a_m\ot b_m\ot\b^{m-1}.
\end{align*}
Then $\Phi(T)$ is the sum of the terms in the $\ZZ_m$-orbit of $a_1\ot a_2\otc a_{m-1}\ot a_m$.
But such vectors span $(A^{\ot m})^{\BZ_m}$.
In summary:

\begin{proposition}
If $N,n\geq m$, then $\tspan\left(\BP\overline{\TNS}(C_m,N,n)\right)=(A^{\ot m})^{\BZ_m}$.
\end{proposition}

It follows that rank violations that are not explained by the $\ZZ_m$-symmetry must be of a nonlinear origin.
We will see such a phenomenon in the next section.

\section{Proof of Theorem~\ref{thm:2d}}\label{sec:Nn2}

In this section, we give a geometric construction of the relevant tensor network variety and prove Theorem~\ref{thm:2d}.

\subsection{\texorpdfstring{A geometric construction of $\BP\overline{\TNS}(C_{2d},2,2)$}{A geometric construction of PTNS(C\_2d,2,2)}}
Let $A=B=\bbC^2$, let~$a_0$, $a_1$ be a basis of~$A$ with dual basis~$\alpha^0$, $\alpha^1$ of~$A^*$, and $b_0$, $b_1$ a basis of~$B$ with dual basis~$\beta^0$, $\beta^1$ of~$B^*$.

Consider the image under $\Phi$ of the line
\begin{align*}
  L[\mu,\nu]:=\left[\mu(a_0\ot b_0\ot\b^0+a_1\ot b_1\ot \b^1)+\nu(a_0\ot b_0\ot\b^1+a_1\ot b_1\ot\b^0)\right] \subseteq \BP(A\ot B\ot B^*)=\pp7.
\end{align*}
In \S\ref{sec:geometry} we discussed the special cases~$\nu=0$ and $\mu=0$. When $\nu=0$, we obtain the point $[a_0^{\ot 2d}+a_1^{\ot 2d}]$, which is a general point of the variety $\s_2(v_{2d}(\BP A)) \subseteq \BP(S^{2d}A)$, which has dimension~$3$. When $\mu=0$, we obtain $[a_0^{\ot d} \ot a_1^{\ot d}+a_1^{\ot d} \ot a_0^{\ot d}] \not\in\BP(S^{2d}A)$. Consider the closure of the union of the $\PGL(A)$-orbits of the points of $\Phi(L[\mu,\nu])$, that is $\bar{\PGL(A) \cdot \Phi(L[\mu,\nu])}$: this is an irreducible projective variety of dimension $4$ sitting inside $\BP \overline{\TNS}(C_{2d},2,2)$ because $\dim\PGL(A)=3$ and the stabilizer of $[a_0^{\ot 2d}+a_1^{\ot 2d}]$ in $\PGL(A)$ is finite.

But $\Phi$ is a rational map from $\BP(A\ot B\ot B^*)=\pp7$ whose fibers have dimension at least $\dim\PGL(B)=3$.
Thus the Zariski closure of its image $\BP\overline{\TNS}(C_{2d},2,2)$ is an irreducible variety of dimension at most $4$, so $\BP \overline{\TNS}(C_{2d},2,2)=\overline{\PGL(A)\cdot \Phi(L[\mu,\nu])}$.
We summarize:

\begin{proposition}
  The variety $\BP\overline{\TNS}(C_{2d},2,2)$ coincides with the variety $\overline{\PGL(A) \cdot \Phi(L[\mu,\nu])}$.
\end{proposition}

We adopt the following notation: $\udelta$ denotes a $2d$-tuple of elements in~$\{0,1\}$, and $\ueta$ and $\ueps$ denote $d$-tuples of elements in $\{0,1\}$. We use odd indices for the entries of $\ueta$ and even indices for the entries of $\ueps$. We write $\ueta \smile \ueps$ for the $2d$-tuple obtained by interlacing $\ueta$ and $\ueps$. Thus, if $\ueta = (\eta_1,\eta_3 \vvirg \eta_{2d-1})$ and $\ueps = (\eps_2,\eps_4 \vvirg \eps_{2d})$ then~$\udelta = \ueta \smile \ueps = (\eta_1 , \eps_2,\eta_3 \vvirg \eps_{2d})$. We write $a_{\udelta}$ etc.\ for the corresponding basis vectors. For example, the contraction of $a_{\ueta \smile \ueps}$ by $\alpha^\ueta$ is $a_\ueps$. All indices are to be read modulo~$2d$.

For every $\udelta \in \{ 0 ,1\}^{2d}$, define
\begin{align}\label{eq:def coeff}
\coef(\udelta) = \sum_{i=1}^{2d} \delta_{i} \boxplus \delta_{i+1},
\end{align}
where $\boxplus$ denotes the XOR operation, i.e., $\delta_{i} \boxplus \delta_{i+1}=1$ if $\delta_{i}\neq  \delta_{i+1}$ and $\delta_{i} \boxplus \delta_{i+1}=0$ otherwise. The value $\coef(\udelta)$ counts the number of changes between $0$ and $1$ that one observes cyclically reading $\udelta$.
Directly from~\eqref{eq:contraction in coordinates}, we observe that $\coef(\udelta)$ determines the coefficient of $\udelta$ in $\Phi(L(\mu,\nu))$.
Explicitly,
\begin{align}\label{eq:formula for F}
  \Phi(L(\mu, \nu))
= \sum_{\udelta} \mu^{2d - \coef(\udelta)} \nu^{\coef(\udelta)} \, a_\udelta
= \sum_c \biggl[\mu^{2d-c}\nu^c \sum_{\coef(\udelta)=c} a_\udelta\biggr].
\end{align}
Moreover $\coef(\udelta)$ is always even. See Figure~\ref{fig:ladder} for a schematic representation of the contribution of $\mu$ or $\nu$ in the contraction:

\begin{figure}[htp!]
\begin{center}
\includegraphics{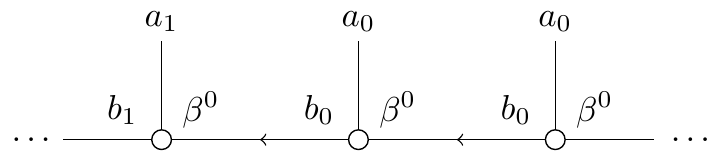}
\end{center}
\caption{This segment produces the monomial $\nu\mu\mu$.}\label{fig:ladder}
\end{figure}

\begin{remark}
The image of $L[\mu,\nu]$ under $\Phi$ is a rational normal curve of degree $d$, as opposed to the na\"\i ve $2d$.
Indeed, since only even powers of $\mu$ and $\nu$ appear, one can reparametrize the image setting $\mu'=\mu^2,\nu'=\nu^2$, showing that it is a rational curve of degree at most $d$.
Moreover, observe that for every even $c = 0,2 \vvirg 2d$ there exists a $d$-tuple $\udelta$ such that $\coef(\udelta) = c$ and the corresponding $a_\udelta$'s are linearly independent.
This shows that $\Phi(L[\mu,\nu])$ spans a $d$-dimensional subspace of $\bbP(A^{\otimes 2d})$, which guarantees that it is a normal curve and that its degree is (at least) $d$.
\end{remark}

\subsection{Proof of Theorem~\ref{thm:2d}}
Let $F:=  (\Phi(L(\mu,\nu)))_{\odd,\even}$ be the flattening of interest for some general choice of $[\mu,\nu] \in \pp1$.
Since the rank of the flattening is $\GL(A)$-invariant, it is sufficient to prove $\trank(F) \leq \frac{3}{4} 2^{d}$, as $\Phi(L[\mu,\nu])$ is a general point of $\BP\overline{\TNS}(C_{2d},2,2)$.

We describe a set of elements in the kernel of~$F$.
Let $S \subseteq \odd=\{1,3,\dots,2d-1\}$ be a non-empty subset of even cardinality, $\vert S \vert = 2p$.
Let $D_S = (\alpha^0 \otimes \alpha^1)^{\ot p} - (\alpha^1 \otimes \alpha^0)^{\ot p} \in (A^*)^{\ot S}$ and let
\begin{align}\label{eq:def K_S}
K_S = D_S \otimes (\alpha^0 - \alpha^1)^{\ot (d - 2p)} \in (A^*)^{\ot S} \otimes (A^*)^{\ot (\odd \setminus S)} = (A^*)^{\ot\odd}.
\end{align}
In this section, we will prove that $K_S \in \ker(F)$ for every $S$ and that the $K_S$'s span a subspace of dimension at least $2^{d-2}$.
This provides $\trank(F) \leq 2^d - 2^{d-2} =\frac{3}{4} 2^{d}$ and therefore $\QMF'(C_{2d},(\odd,\even),2) \leq  \frac{3}{4} 2^d$, establishing Theorem~\ref{thm:2d}.

For any fixed $S\subseteq\odd$, the basis vectors $\alpha^\ueta$ in the support of~$K_S$ are labeled by~$d$-tuples~$\ueta$ that are in one-to-one correspondence with elements in $\{0,1\}\times\{0,1\}^{d-2p}$ as follows (here $\ueta'$ denote a $(d-2p)$-tuple):
\begin{align*}
  \{0,1\}\times\{0,1\}^{d-2p} &\to \{0,1\}^S \times \{0,1\}^{\odd\setminus S} \cong \{0,1\}^\odd \\
  (0,\ueta') &\mapsto ((0,1,\dots,0,1), \ueta'), \\
  (1,\ueta') &\mapsto ((1,0,\dots,1,0), \ueta').
\end{align*}
We denote the elements of $\{0,1\}\times\{0,1\}^{d-2p}$ by $\ueta=(t,\ueta')$ and identify them with their image.
The coefficient of $\alpha^\ueta$ in $K_S$ is $\sgn(\ueta):=(-1)^{t + \lvert\ueta'\rvert}$, where $\lvert\ueta'\rvert$ is the sum of the entries of~$\ueta'$, as can be readily seen from~\eqref{eq:def K_S}.
Thus, $K_S = \sum_{\ueta} \sgn(\ueta) \alpha^\ueta$, and if we plug this into~\eqref{eq:formula for F} then we obtain
\begin{align*}
  F(K_S)
= \sum_{\ueps} \left(\sum_{\ueta} \sgn(\ueta) \mu^{2d - \coef(\ueta\smile\ueps)} \nu^{\coef(\ueta\smile\ueps)} \right) a_\ueps.
\end{align*}
Thus we can read off the following criterion:

\begin{lemma}\label{lem:zero if involution}
  The coefficient of $a_\ueps$ in $F(K_S)$ is zero if there exists a permutation $\Theta_\ueps$ of $\{0,1\}\times\{0,1\}^{d-2p}$ such that, for every $\ueta$, $\coef(\ueta\smile\ueps)=\coef(\Theta_\ueps(\ueta)\smile\ueps)$ while $\sgn(\Theta_\ueps(\ueta))=-\sgn(\ueta)$.
\end{lemma}

We now show that such permutations always exist, which proves that each $K_S$ is in the kernel.

\begin{proposition}\label{prp:K_S in kernel}
For every subset $S\subseteq\odd=\{1,3,\dots,2d-1\}$ of even cardinality, and for every $\ueps=(\eps_2,\dots,\eps_{2d})$, there exists a permutation $\Theta_\ueps$ as in Lemma~\ref{lem:zero if involution}.
As a consequence, each $K_S\in\ker(F)$.
\end{proposition}
\begin{proof}
We distinguish two cases.
First, assume that there exists an odd index $k\not\in S$ such that $\eps_{k-1}\neq\eps_{k+1}$.
Define an involution $\Theta_\ueps$ that replaces $\eta_k$ by its complement $\eta_k\boxplus1$ (i.e., if $\eta_k=0$, it is replaced by 1, and vice versa) while leaving all other elements the same.
Then $\sgn(\Theta_\ueps(\ueta))=-\sgn(\ueta)$ for every $\ueta$.
Moreover, $\coef(\Theta_\ueps(\ueta))=\coef(\ueta)$ because the only two terms in the summation~\eqref{eq:def coeff} that involve $\eta_k$ are
\begin{align*}
  (\eps_{k-1} \boxplus \eta_k) + (\eta_k \boxplus \eps_{k+1}) = (0 \boxplus \eta_k) + (1 \boxplus \eta_k) = 1,
\end{align*}
so their sum is independent of the value of $\eta_k$.

Now assume that $\eps_{k-1}=\eps_{k+1}$ for all odd $k\not\in S$.
Define an involution $\Theta_\ueps(t,\ueta') = (t\boxplus1,\ueta')$, i.e., every entry~$\eta_k$ for $k\in S$ is reversed while all other entries are unchanged.
Again, it is clear that $\sgn(\Theta_\ueps(\ueta))=-\sgn(\ueta)$ for every $\ueta$.
We now argue that $\coef(\Theta_\ueps(\ueta))=\coef(\ueta)$.
Let $S=\{k_1<\dots<k_{2p}\}$.
Then the only terms in the summation~\eqref{eq:def coeff} that involve indices in~$S$ are
\begin{align}\label{eq:relevant sum}
\sum_{j=1}^{2p} \left( (\eta_{k_j} \boxplus \eps_{k_j+1}) + (\eps_{k_{j+1}-1} \boxplus \eta_{k_{j+1}}) \right)
\end{align}
where we set $k_{2p+1}:=k_1$.
Note that $\eps_{k_j+1}=\eps_{k_{j+1}-1}$ (for $k_{j+1}=k_j+2$ this is trivial, otherwise use the assumption).
Moreover, $\eta_{k_j} = 1 - \eta_{k_{j+1}}$, which remains true when we apply $\Theta_\ueps$.
As a consequence, applying $\Theta_\eps$ only exchanges 0 and 1 in each summand of~\eqref{eq:relevant sum}, and the total is preserved.
\end{proof}

To conclude the proof of Theorem~\ref{thm:2d}, it remains to show that $\bfK =  \tspan \{ K_S : \emptyset \neq S \subseteq \odd, \lvert S \rvert \text{ even} \}$ is of dimension at least $2^{d-2}$.
In fact, we will prove that $\bfK_{(1)} = \tspan \{ K_S : 1\in S\subseteq \odd, \lvert S\rvert \text{ even} \}$ has dimension equal to $2^{d-2}$.
Note that there are exactly $2^{d-2}$ subsets $S$ of $\{1,3,\dots,2d-1\}$ with even cardinality that contain 1.
Therefore, we need to show that the corresponding $K_S$ are linearly independent.
We start with the following lemma:

\begin{lemma}\label{lem:K decomposition}
We have $\bfK_{(1)} = \bfK_{(1,3)} \op \bfK_{(1,\bar 3)}$, where $\bfK_{(1,3)} := \tspan \{ K_S : 1,3\in S \subseteq\odd, \lvert S\rvert\text{ even} \}$ and $\bfK_{(1,\bar 3)} := \tspan \{ K_S : 1\in S \subseteq\odd, 3\not\in S, \lvert S\rvert\text{ even} \}$.
Further, $\bfK_{(1,3)} = \bfK_{(1,3,5)} \op \bfK_{(1,3,\bar 5)}$, where $\bfK_{(1,3,5)} := \tspan \{ K_S : 1,3,5\in S \subseteq\odd, \lvert S\rvert\text{ even} \}$ and $\bfK_{(1,3,\bar 5)} := \tspan \{ K_S : 1,3\in S \subseteq\odd, 5\not\in S, \lvert S\rvert\text{ even} \}$.
\end{lemma}
\begin{proof}
Clearly $\bfK_{(1,3)}$ and $\bfK_{(1,\bar3)}$ generate $\bfK_{(1)}$.
It remains to show that their intersection is trivial.
Say it were not and let $v\in\bfK_{(1,3)} \cap \bfK_{(1,\bar3)}$ be nonzero.
Choose generic projections $A^* \to \bbC$ on the factors $5 \vvirg 2d-1$ of $(A^*)^{\ot\odd}$, and consider the image of $v$ in $A^* \otimes A^*$.
On the one hand, since $v \in \bfK_{(1,3)} $, the image is of the form $x \alpha^0 \otimes \alpha^1 + y \alpha^1 \otimes \alpha^0$
By the genericity of the projections, we may assume that both $x$ and $y$ are nonzero, so that the tensor has rank two.
On the other hand, since $v \in \bfK_{(1,\bar3)}$, the image is of the form of $(z \alpha^0 + w \alpha^1) \otimes (\alpha^1 - \alpha^0)$, a rank-one tensor.
This is a contradiction.
The second statement is proved analogously.
\end{proof}

We now show that $\bfK_{(1)}$ has the desired dimension.

\begin{lemma}\label{lem:dim K1}
For $d\geq2$, we have that $\dim \bfK_{(1)} = 2^{d-2}$.
\end{lemma}
\begin{proof}
We will prove that $\dim \bfK^{(d)}_{(1)} = 2^{d-2}$ and $\dim \bfK^{(d)}_{(1,3)} = 2^{d-3}$ by induction on $d$.
The base cases are as follows:
$\bfK^{(2)}_{(1)}=\langle K_{\{1,3\}} \rangle$,
$\bfK^{(3)}_{(1)}=\langle K_{\{1,3\}}, K_{\{1,5\}} \rangle$, and
$\bfK^{(3)}_{(1,3)}=\langle K_{\{1,3\}} \rangle$
(recall we only consider nonempty subsets of even cardinality).

Now let $d\geq4$.
Using Lemma~\ref{lem:K decomposition}, we have
\begin{align*}
\bfK^{(d)}_{(1)}
= \bfK^{(d)}_{(1,3)} \op \bfK^{(d)}_{(1,\bar3)}, \quad
\bfK^{(d)}_{(1,3)}
= \bfK^{(d)}_{(1,3,5)} \op \bfK^{(d)}_{(1,3,\bar5)}.
\end{align*}
Note that $\bfK^{(d)}_{(1,\bar 3)} \subseteq (\alpha^0-\alpha^1)\ot (A^*)^{\ot(\odd\setminus\{3\})}$, so we see that $\bfK^{(d)}_{(1,\bar 3)} \cong \bfK^{(d-1)}_{(1)}$.
Similarly, $\bfK^{(d)}_{(1,3,\bar 5)} \subseteq (\alpha^0-\alpha^1)\ot (A^*)^{\ot(\odd\setminus\{5\})}$ and $\bfK^{(d)}_{(1,3,\bar 5)} \cong \bfK^{(d-1)}_{(1,3)}$.
Finally, note that $\bfK^{(d)}_{(1,3,5)}\cong\bfK^{(d-2)}_{(1)}$.
Indeed, both spaces have the same number of generators; contraction with $a_0\ot a_1 + a_1\ot a_0 \in A^{\ot\{1,3\}}$ maps each generator $K_S$ onto a generator $K_{S'}$, where $S' = \{ k - 4 : k \in S, k\neq1,3\}$; all the latter are distinct and therefore linearly independent by the induction hypothesis.
Thus:
\begin{align*}
  \bfK^{(d)}_{(1)} \cong \bfK^{(d)}_{(1,3)} \op \bfK^{(d-1)}_{(1)}, \quad
  \bfK^{(d)}_{(1,3)} \cong \bfK^{(d-2)}_{(1)} \op \bfK^{(d-1)}_{(1,3)},
\end{align*}
and so
\begin{align*}
\dim\bfK^{(d)}_{(1,3)} = 2^{(d-2)-2} + 2^{(d-1)-3} = 2^{d-3}, \quad
\dim\bfK^{(d)}_{(1)} = 2^{d-3} + 2^{(d-1)-2} = 2^{d-2}
\end{align*}
using the induction hypothesis.
This concludes the proof.
\end{proof}

Proposition~\ref{prp:K_S in kernel} and Lemma~\ref{lem:dim K1} together establish Theorem~\ref{thm:2d}.\qed{}

\begin{acknowledgments}
We thank Shrawan Kumar for a suggestion regarding the proof of Lemma~\ref{lem:mults}.

Gesmundo supported by the European Research Council (ERC Grant Agreement no.~337603), the Danish Council for Independent Research (Sapere Aude), and VILLUM FONDEN via the QMATH Centre of Excellence (grant no.~10059).
Landsberg supported by NSF DMS-1405348 and NSF CCF-1814254.
Walter acknowledges support by the Simons Foundation, AFOSR (grant no.~FA9550-16-1-0082), and the NWO (grant no.~680-47-459).
\end{acknowledgments}

\bibliography{GLWqmf}

\end{document}